\documentclass[12pt,a4paper]{article}
\usepackage{amsmath,amsthm,amssymb,amscd}
\usepackage{latexsym}
\usepackage{indentfirst}
\usepackage{bibentry}
\usepackage{textcomp}
\usepackage{float}
\usepackage{multirow}
\usepackage{booktabs}
\usepackage{graphicx}
\usepackage{color}
\usepackage{cite}
\usepackage{authblk}
\usepackage[numbers]{natbib}
\setcitestyle{open={},close={}}
\makeatletter \@addtoreset{equation}{section}

\makeatother
\newtheorem{thm}{Theorem}[section]
\newtheorem{cor}{Corollary}[section]
\newtheorem{lem}{Lemma}[section]

\theoremstyle{definition}
\newtheorem{rem}{Remark}[section]

\begin{document}

\title{\Large {\color{black}The smallest} eigenvalue of large Hankel matrices generated by a singularly perturbed Laguerre weight}

\author[1]{{\large Mengkun Zhu}\footnote{Zhu\_mengkun@163.com}}
\author[2]{{\large Yang Chen}\footnote{yayangchen@um.edu.mo}}
\author[3]{{\color{black}{\large Chuanzhong Li}\footnote{\color{black}Corresponding author:lichuanzhong@nbu.edu.cn}}}
\affil[1]{\large{School of Mathematics and Statistics}, Qilu University of Technology (Shandong Academy of Sciences)\\
Jinan 250353, China}
\affil[2]{\large Department of Mathematics, University of Macau,
Avenida da Universidade, Taipa, Macau, China}
\affil[3]{\color{black}\large School of Mathematics and Statistics, Ningbo University,
Ningbo 315211, China}

\renewcommand\Authands{ and }

\date{}
\maketitle

\begin{abstract}
An asymptotic expression of the orthonormal polynomials $\mathcal{P}_{N}(z)$ as $N\rightarrow\infty$, associated with the singularly perturbed Laguerre weight $w_{\alpha}(x;t)=x^{\alpha}{\rm e}^{-x-\frac{t}{x}},~x\in[0,\infty),~\alpha>-1,~t\geq0$ is derived. Based on this, we establish the asymptotic behavior of the smallest eigenvalue, $\lambda_{N}$, of the Hankel matrix generated by the weight $w_{\alpha}(x;t)$.
\end{abstract}


\section{Introduction}
\noindent The analysis of Hankel matrices, occurs naturally in moment problems, which plays an important role in random matrix theory. The smallest eigenvalue is important since it gives great useful information of the nature of Hankel matrices with respect to given weights.

The moment sequence generated by a weight function $w(x)$ is given by
\begin{equation}\label{b2}
\mu_{j}:=\int x^{j}w(x)dx,~~j=0,1,2,\ldots,
\end{equation}
it is known that the associated Hankel matrices,
\begin{equation}\label{b1}
\mathcal{H}_{N}:=\left(\mu_{j+k}\right)_{j,k=0}^{N},~~ N=0,1,2,\ldots
\end{equation}
are positive definite.

Let $\lambda_{N}$ denote the smallest eigenvalue of $\mathcal{H}_{N}$. The asymptotic behavior of $\lambda_{N}$ for large $N$ has been widely investigated in recent decades. Szeg\"{o} [\cite{C3}] studied the asymptotic behavior of $\lambda_{N}$ for the Hermite weight ($w(x)={\rm e}^{-x^{2}},x\in(-\infty,\infty)$) and the Laguerre weight ($w(x)={\rm e}^{-x},x\geq0$). He found\footnote[1]{ In all of this paper, $a_{N}\sim b_{N}$ means $\lim_{N\rightarrow\infty}a_{N}/b_{N}$=1.}
\begin{equation*}
\lambda_{N}\sim A N^{\frac{1}{4}}B^{\sqrt{N}},
\end{equation*}
where $A,B$ are certain constants with satisfying $0<A,~0<B<1$.

Widom and Wilf [\cite{C6}] investigated the case of $w(x)$ on a compact interval $[a,b]$ with the Szeg\"{o} condition holds, they obtained
\begin{equation*}
\lambda_{N}\sim A \sqrt{N}B^{N}.
\end{equation*}

Chen and Lawrence [\cite{C9}] found the asymptotic behavior of $\lambda_{N}$ with the weight function $w(x)={\rm e}^{-x^{\beta}},~x\in[0,\infty),~\beta>\frac{1}{2}$. Berg, Chen and Ismail [\cite{C13}] proved that the moment sequence (\ref{b2}) is determinate iff $\lambda_{N}\rightarrow0$ as $N\rightarrow\infty$ which is a new criteria for the determinacy of the Hamburger moment problem. In [\cite{C10}], Chen and Lubinsky deduced out the behavior of $\lambda_{N}$ when $w(x)={\rm e}^{-|x|^{\alpha}},~x\in(-\infty, \infty),~\alpha>1$. Berg and Szwarc [\cite{C14}] proved that $\lambda_{N}$ has exponential decay to zero for any measure which with compact support.

In recent several years, Zhu, etc. studied the {\color{black}Jacobi} case [\cite{C22}], i.e. $w(x)=x^{\alpha}(1-x)^{\beta},~x\in[0,1],~\alpha>-1,~\beta>-1$; the generalized Laguerre weight [\cite{MMAS}] $w(x)=x^{\alpha}{\rm e}^{-x^{\beta}},~x\in[0,\infty),~\alpha>-1,~\beta>\frac{1}{2}$; as well as the weight function $w(x)={\rm e}^{-x^{\beta}},~x\in[0,\infty)$ at the critical point $\beta=\frac{1}{2}$, see [\cite{AMC2}].

In this paper, we choose the singularly perturbed Laguerre weight $w_{\alpha}(x;t)=x^{\alpha}{\rm e}^{-x-\frac{t}{x}},~x\in[0,\infty),~\alpha>-1,~t\geq0$, which was motivated in part by an integrable quantum field theory at finite temperature. It transpires that this is equivalent to the characterization of a sequence of polynomials orthogonal with respect to the weight $w_{\alpha}(x;t)$. Some earlier researches on this weight, e.g., the weight $w_{\alpha}(x;t)$ and Painlev\'{e} \textrm{III}, please see [\cite{C-C,chenits}].

In the next section we analysis the properties of Hankel matrices generated by the above weight by using the modified Bessel function. We also reproduce some known results that will be applied to find the estimation of $\lambda_{N}$. In section 3, by adopting a previous result [\cite{C20}], we obtain the asymptotic formula for the polynomials orthonormal with respect to $w_{\alpha}(x;t)=x^{\alpha}{\rm e}^{-x-\frac{t}{x}},~x\in[0,\infty),~\alpha>-1,~t\geq0$, which is then employed in sections 4 for the determination of the large $N$ behavior of $\lambda_{N}$.

\section{Preliminaries}
Consider the singularly perturbed Laguerre weight
\begin{equation*}
w(x):=x^{\alpha}{\rm e}^{-x-\frac{t}{x}},~x\in[0,\infty),~\alpha>-1,~t\geq0.
\end{equation*}
The Hankel matrix, generated by the rapidly vanishing weight mentioned above is given by
\begin{equation*}
\left(\mu_{j+k}(t)\right)=\left(\int_{0}^{\infty}x^{j+k}x^{\alpha}{\rm e}^{-x-\frac{t}{x}}\right)_{0\leq j,k\leq N},~~\alpha>-1,~~t\geq0.
\end{equation*}

If $t=0$, the weight function, i.e., the factor to the right of $x^{j+k}$ is the Laguerre or the gamma density. The matrix elements is actually the integral representation of the modified Bessel function of the second kind {\color{black}[\cite{new2}, p. 917]}, since
\begin{equation*}
\int_{0}^{\infty}x^{k}x^{\alpha}{\rm e}^{-x-\frac{t}{x}}dx=2t^{\frac{\alpha+k+1}{2}}\mathrm{K}_{\alpha+k+1}(2\sqrt{t}),~~t>0.
\end{equation*}
We could of course take a special value of $t$.

For a fixed $t(>0)$, and $\nu\rightarrow\infty$ through positive values,
\begin{equation*}
\mathrm{K}_{\nu}(t)\sim\frac{1}{2}\frac{\Gamma(\nu)}{(t/2)^{\nu}},
\end{equation*}
so the moments, for large and positive $\alpha+k+1$ becomes,
\begin{equation*}
\Gamma(\alpha+k+1).
\end{equation*}

{\color{black}From the Laplace's method, it is evident that the main terms in the asymptotic expansion of
\begin{equation*}
\int_{0}^{\infty}x^{k}x^{\alpha}{\rm e}^{-x-\frac{t}{x}}dx
\end{equation*}
will not involve $t$. We write the integral as $\int_{0}^{\infty}{\rm e}^{-v(x)}dx$ and with the exponent
\begin{equation*}
v(x)=x+\frac{t}{x}-(\alpha+k)\log x
\end{equation*}
we get
\begin{equation*}
v'(x)=1-\frac{t}{x^2}-\frac{\alpha+k}{x}
\end{equation*}
and
\begin{equation*}
v''(x)=\frac{2t}{x^3}+\frac{\alpha+k}{x^2}.
\end{equation*}
The equation $v'(x_0)=0$ has solution
\begin{equation*}
x_0=\frac{\alpha+k+\sqrt{(\alpha+k)^2-4t}}{2}=\alpha+k+O\left(\frac{t}{k}\right),~~k\rightarrow\infty,
\end{equation*}
so the leading terms in $v(x_0)$ and $v''(x_0)$ do not depend upon $t$.
}We see this problem is similar to the previous problem studied by Emmart, Chen and Weems [\cite{C11}].

In particular, with $\nu=n+\frac{1}{2},~n=0,1,\ldots$, $\mathrm{K}_{\nu}$ arises in quantum chemistry, for such $\nu$, i.e., for $\alpha=n-\frac{1}{2},~n\geq0$, it is known as the spherical Bessel functions. The relevant paper on $\mathrm{K}_{\nu}(t)$ for $t$ fixed and large $\nu$ can be found from [\cite{Sidi}].
The focus of this paper is to derive the asymptotic behavior of the smallest eigenvalue $\lambda_{N}$ of $\mathcal{H}_{N}$.

It is well known that the smallest eigenvalue $\lambda_{N}$ can be found using the classical Rayleigh quotient
\begin{equation}\label{e1}
\lambda_{N}=\min\left\{\frac{\sum_{j,k=0}^{N}\overline{x}_{j}\mu_{j+k}x_{k}}{\sum_{k=0}^{N}|x_{k}|^{2}}~\Bigg{|}~X:=\left(x_{0},x_{1},\ldots,x_{N}\right)^{T}\in\mathbb{C}^{N+1}
\setminus\{0\}\right\}.
\end{equation}
Let $P_{N}$ be the orthogonal polynomials associated with $w(x)$, and denote by
\begin{equation*}
P_{N}(z):=\sum_{k=0}^{N}x_{k}z^{k},
\end{equation*}
then
\begin{equation}\label{0.00}
\int_{0}^{\infty}|P_{N}(x)|^{2}w(x)dx=\sum_{j,k=0}^{N}\overline{x}_{j}\mu_{j+k}x_{k},
\end{equation}
Consequently, in the condition of
\begin{equation}\label{200}
\int_{0}^{\infty}|P_{N}(x)|^{2}w(x)dx=1,
\end{equation}
the expression for $\lambda_{N}$, (\ref{e1}), can be recast as {\color{black}[\cite{C3}, p.453]}
\begin{equation}\label{e2}
\lambda_{N}=\min\left\{\frac{2\pi}{\int_{-\pi}^{\pi} \left|P_{N}({\rm e}^{{\rm i}{\color{black}\phi}})\right|^{2}d{\color{black}\phi}} \right\}.
\end{equation}

We can also obtain the polynomials denoted by $\mathcal{P}_{N}(z)$, orthonormal with respect to $w(x)$, through
\begin{equation*}
P_{N}(z)=\sqrt{h_{N}}\mathcal{P}_{N}(z),
\end{equation*}
where $h_{N}$ is the square of the $L^{2}$ norm of $P_{N}(z)$. If we define
\begin{equation*}
P_{N}(z):=\sum_{k=0}^{N}\xi_{k}\mathcal{P}_{k}(z),~~~{\rm and}~~~\mathcal{K}_{jk}:=\int_{-\pi}^{\pi}\mathcal{P}_{j}\left({\rm e}^{{\rm i}{\color{black}\phi}}\right)\mathcal{P}_{k}\left({\rm e}^{-{\rm i}{\color{black}\phi}}\right)d{\color{black}\phi},
\end{equation*}
then we can see that
\begin{equation*}
\int_{-\pi}^{\pi}\left|P_{N}\left({\rm e}^{{\rm i}{\color{black}\phi}}\right)\right|^{2}d{\color{black}\phi}=\sum_{j,k=0}^{N}\overline{\xi}_{j}\mathcal{K}_{jk}\xi_{k},
\end{equation*}

Again with the condition given by (\ref{200}), i.e.
\begin{equation*}
\sum_{k=0}^{N}\left|\xi_{k}\right|^{2}=1,
\end{equation*}
 the formula (\ref{e2}) will be equivalent to
\begin{equation}\label{e3}
\lambda_{N}=\min\left\{\frac{2\pi}{\sum_{j,k=0}^{N}\overline{\xi}_{j}\mathcal{K}_{jk}\xi_{k}}  \right\}.
\end{equation}
Based on the Cauchy's and Schwarz inequality, we will find that
\begin{equation*}
\begin{split}
\sum_{j,k=0}^{N}\overline{\xi}_{j}\mathcal{K}_{jk}\xi_{n}\leq\sum_{j,k=0}^{N}\mathcal{K}_{jj}^{\frac{1}{2}}\mathcal{K}_{kk}^{\frac{1}{2}}\left|\xi_{j}\right|\left|
\xi_{k}\right|\leq\sum_{j=0}^{N}\mathcal{K}_{jj}\cdot\sum_{k=0}^{N}\left|\xi_{k}\right|^{2}
\leq\sum_{k=0}^{N}\mathcal{K}_{kk}.
\end{split}
\end{equation*}
Therefore, a lower bound for the smallest eigenvalue of $\lambda_{N}$ is given by
\begin{equation}\label{e4}
\lambda_{N}\geq\frac{2\pi}{\sum_{k=0}^{N}\mathcal{K}_{kk}}.
\end{equation}

\section{The orthonomal polynomials with respect to the singularly perturbed Laguerre weight.}
\subsection{The preliminary expression of orthonomal polynomials}
{\color{black}We first give a brief description of Coulomb fluid [\cite{C20}]. The energy of a system of $N$ logarithmically repelling particles on the line confined by an external potential $v(x)$ reads
\begin{equation*}
E\left(x_{1},x_{2},\ldots,x_{N}\right)=-2\sum_{1\leq j<k\leq N}\log\left|x_{j}-x_{k}\right|+\sum_{j=1}^{N}v\left(x_{j}\right),
\end{equation*}
The collection particles, for large enough $N$, can be approximated as a continuous fluid with a certain density $\sigma(x)$ supported on a single interval $[a,b]$ [\cite{Dyson}]. This density, $\sigma(x)$, corresponding to the equilibrium density of the fluid, is obtained by the constrained minimization of the free-energy function, $F[\sigma]$,
\begin{equation*}
F[\sigma]:=\int_{a}^{b}\sigma(x)v(x)dx-\int_{a}^{b}\int_{a}^{b}\sigma(x)\log|x-y|\sigma(y)dxdy.
\end{equation*}
subject to
\begin{equation*}
\int_{a}^{b}\sigma(x)dx=N,~~~\sigma(x)>0.
\end{equation*}

Upon minimization [\cite{Tsuji}], the equilibrium density $\sigma(x)$ is found to satisfy the integral equation,
\begin{equation*}
L:=v(x)-2\int_{a}^{b}\log|x-y|\sigma(y)dy,~~~~x\in[a,b],
\end{equation*}
where $L$ is the Lagrange multiplier. The derivative of this equation over $x$ gives rise to a singular integral equation
\begin{equation*}
v'(x)-2P\int_{a}^{b}\frac{\sigma(y)}{x-y}dy=0,~~~x\in[a,b].
\end{equation*}
where $P$ denotes the Cauchy principal value. Based on the theory of singular integral equations [\cite{Mikhlin}], we find
\begin{equation*}
\sigma(x)=\frac{1}{2\pi^{2}}\sqrt{\frac{b-x}{x-a}}P\int_{a}^{b}\frac{v'(y)}{y-x}\sqrt{\frac{y-a}{b-y}}dy.
\end{equation*}
Thus the normalization, $\int_{a}^{b}\sigma(x)dx=N$, becomes
\begin{equation*}
\frac{1}{2\pi}\int_{a}^{b}\sqrt{\frac{y-a}{b-y}}v'(y)dy=N.
\end{equation*}
with a supplementary condition [\cite{C20}],
\begin{equation*}
\int_{a}^{b}\frac{v'(x)}{\sqrt{(b-x)(x-a)}}=0.
\end{equation*}
Consequently, the normalization condition becomes
\begin{equation*}
\int_{a}^{b}\frac{xv'(x)}{\sqrt{(b-x)(x-a)}}=2\pi N.
\end{equation*}
For the problem at hand, the potential, $v(x)=-\alpha\log x+x+\frac{t}{x}$ is convex. Hence, applying the results of [\cite{C20}, (4.2)-(4.8)] directly, i.e.} the monic orthogonal polynomials $P_{N}(z)$ associated with $w(x)={\rm e}^{-v(x)}$, can be approximated by
\begin{equation}\label{pn}
P_{N}(z)\sim\exp\left[-S_{1}(z)-S_{2}(z)\right],
\end{equation}
where
\begin{equation*}
S_{1}(z)=\frac{1}{4}\ln\left[\frac{16(z-a)(z-b)}{(b-a)^{2}}\left(\frac{\sqrt{z-a}-\sqrt{z-b}}{\sqrt{z-a}+\sqrt{z-b}}\right)^{2}\right],~\text{$z\notin [a,b]$},
\end{equation*}

\begin{equation*}
\begin{split}
S_{2}(z)  &=-N\ln\left(\frac{\sqrt{z-a}+\sqrt{z-b}}{2}\right)^{2}\\
         &  +\frac{1}{2\pi}\int_{a}^{b}\frac{v(x)}{\sqrt{(b-x)(x-a)}}\left[\frac{\sqrt{(z-a)(z-b)}}{x-z}+1\right]dx,~\text{$z\notin [a,b]$}.
\end{split}
\end{equation*}
Chen and his co-authors also gave an equivalent representation for $S_{1}$:
\begin{equation*}
{\rm e}^{-S_{1}(z)}=\frac{1}{2}\left[\left(\frac{z-b}{z-a}\right)^{\frac{1}{4}}+\left(\frac{z-a}{z-b}\right)^{\frac{1}{4}}\right],~\text{$z\notin [a,b]$}.
\end{equation*}

\begin{lem}\label{lem1}
{\rm [\cite{C20}]} The orthonomal polynomials $\mathcal{P}_{N}(z)$ with respect to the weight $w(x)$, i.e.
\begin{equation*}
\int_{a}^{b}\left[\mathcal{P}_{N}(x)\right]^{2}w(x)dx=1,
\end{equation*}
can be given by:
\begin{equation*}
\mathcal{P}_{N}(z)=\sqrt{\frac{2}{\pi(b-a)}}\exp\left[\frac{A}{2}\right]P_{N}(z),
\end{equation*}
where
\begin{equation*}
A:=2\int_{a}^{b}\frac{v(x)dx}{2\pi\sqrt{(b-x)(x-a)}}-2N\log\left(\frac{b-a}{4}\right),
\end{equation*}
and the orthogonal polynomials $P_{N}(z)$ is approximated by (\ref{pn}).
\end{lem}
From above, one finds that
\begin{lem}
The orthonomal polynomials $\mathcal{P}_{N}(z)$ associated with the weight $w_{\alpha}(x;t)$ can be approximated by
\begin{equation}\label{eqt6}
\begin{split}
\mathcal{P}_{N}(z):=&\frac{1}{\sqrt{2\pi(b-a)}}\cdot\left(\frac{\sqrt{z-a}+\sqrt{z-b}}{\sqrt{b-a}}\right)^{2N}
\cdot\left[\left(\frac{z-b}{z-a}\right)^{\frac{1}{4}}+\left(\frac{z-a}{z-b}\right)^{\frac{1}{4}}\right]\\
&\cdot\exp\left[-\frac{\sqrt{(z-a)(z-b)}}{2\pi}\int_{a}^{b}\frac{v(x)}{(x-z)\sqrt{(b-x)(x-a)}}dx\right].
\end{split}
\end{equation}
\end{lem}

For our problem
\begin{equation*}
v(x)=-\alpha\log x+x+\frac{t}{x},~x\in[0,\infty),~\alpha>-1,~t\geq0,
\end{equation*}
To determine the asymptotic expansion of the orthonomal polynomials $\mathcal{P}_{N}(z)$, for large $N$, and $z\notin(a,b)$, we first handle the below integral,
\begin{equation*}
\begin{split}
I(z):&=\frac{1}{2\pi}\int_{a}^{b}\frac{v(x)}{(x-z)\sqrt{(b-x)(x-a)}}dx\\
&=-\frac{\alpha}{2\pi}\int_{a}^{b}\frac{\log x}{(x-z)\sqrt{(b-x)(x-a)}}dx+\frac{1}{2\pi}\int_{a}^{b}\frac{x}{(x-z)\sqrt{(b-x)(x-a)}}dx\\
&+\frac{1}{2\pi}\int_{a}^{b}\frac{t}{(x^{2}-zx)\sqrt{(b-x)(x-a)}}dx\\
&=\frac{\alpha}{2\sqrt{(z-a)(z-b)}}\log\left[\frac{2ab-(a+b)z+2\sqrt{ab}\sqrt{(a-z)(b-z)}}{\left(\sqrt{a-z}+\sqrt{b-z}\right)^{2}}\right]+\frac{1}{2}\\
&-\frac{z}{2\sqrt{(z-a)(z-b)}}-\frac{t}{2z\sqrt{ab}}-\frac{t}{2z\sqrt{(z-a)(z-b)}},\\
\end{split}
\end{equation*}
which is obtained by using the integrals in the Appendix. Then with some elementary calculations we see that,
\begin{lem}\label{lem4.3}
For $N\rightarrow\infty$, the normalized polynomials associated with the weight $w_{\alpha}(x;t)$ are approximated by
\begin{equation}\label{eqt2}
\begin{split}
\mathcal{P}_{N}(z)\sim & \frac{1}{\sqrt{2\pi(b-a)}}\cdot
\left(\frac{\sqrt{z-a}+\sqrt{z-b}}{\sqrt{b-a}}\right)^{2N}\cdot\left[\left(\frac{z-b}{z-a}\right)^{\frac{1}{4}}+\left(\frac{z-a}{z-b}\right)^{\frac{1}{4}}\right]\\
&\cdot\exp\left\{-\frac{\alpha}{2}\log\left[\frac{2ab-(a+b)z+2\sqrt{ab}\sqrt{(a-z)(b-z)}}{\left(\sqrt{a-z}+\sqrt{b-z}\right)^{2}}\right]\right\}\\
&\cdot\exp\left[\frac{t}{2z}+\frac{t\sqrt{(z-a)(z-b)}}{2z\sqrt{ab}}+\frac{z}{2}-\frac{\sqrt{(z-a)(z-b)}}{2}\right].\\
\end{split}
\end{equation}
where $z\notin[a,b]$.
\end{lem}

In here, the end points $a$ and $b$ are approximated by
\begin{lem}{\rm [\cite{C-C}]}
\begin{equation}\label{eqt8}
\begin{split}
a=a_{N}:&=\frac{t^\frac{2}{3}}{2(2N+\alpha)^{\frac{1}{3}}}+\frac{\alpha t^\frac{1}{3}}{3(2N+\alpha)^{\frac{2}{3}}}+\frac{\alpha^2}{6(2N+\alpha)}
+\mathcal{O}\left((2N+\alpha)^{-\frac{4}{3}}\right),
\end{split}
\end{equation}
and
\begin{equation}\label{eqt9}
\begin{split}
b=b_{N}:&=2(2N+\alpha)+\frac{3t^\frac{2}{3}}{2(2N+\alpha)^{\frac{1}{3}}}-\frac{\alpha t^\frac{1}{3}}{(2N+\alpha)^{\frac{2}{3}}}-\frac{\alpha^2}{6(2N+\alpha)}
+\mathcal{O}\left((2N+\alpha)^{-\frac{4}{3}}\right).
\end{split}
\end{equation}
\end{lem}

\subsection{The asymptotic expression of $\mathcal{P}_{N}(z),~z\notin[a,b]$.}
{\color{black}
\begin{thm}\label{thmnew1}
Let $\eta:=-\frac{z}{b-a}$, then for large $N$, the normalized polynomials associated with the weight $w_{\alpha}(x;t)$ are approximated by
\begin{equation}\label{z}
\begin{split}
\mathcal{P}_{N}(z)\sim&\frac{(-1)^{N}\eta^{-\frac{1}{4}}}{\sqrt{2\pi(b-a)}}(-z)^{-\frac{\alpha}{2}}{\rm e}^{\frac{z}{2}+\frac{t}{2z}}\cdot\exp\bigg\{(2N+1+\alpha)\log\left(\sqrt{\eta}+\sqrt{\eta+1}\right)\\
&+\left[\frac{b-a}{2}-\frac{(b-a)t}{2z\sqrt{ab}}\right]\sqrt{\eta(\eta+1)}\bigg\}.
\end{split}
\end{equation}
where $z\notin[a,b]$.
\end{thm}
\begin{proof}
Substituting $\eta:=-\frac{z}{b-a}$, $|\eta|\ll1$, into (\ref{eqt2}), with choosing the branch $-(b-a)\eta-a=[(b-a)\eta+a]{\rm e}^{{\rm i}\pi}$ and $-(b-a)\eta-b=[(b-a)\eta+b]{\rm e}^{{\rm i}\pi}$, the result is obtained immediately.
\end{proof}

Furthermore, using the inverse hyperbolic sine and the formula, [\cite{new2}, cf.9.121.26], we find
\begin{equation}\label{w}
\begin{split}
\left(2N+\alpha+1\right)&\log\left(\sqrt{\eta}+\sqrt{1+\eta}\right)\sim\left(2N+\alpha\right)\sqrt{\eta}\cdot{_{2}F_{1}}\left(\frac{1}{2},\frac{1}{2};\frac{3}{2};-\eta\right)\\
&=\left(2N+\alpha\right)\sqrt{\eta}\cdot\sum_{k=0}^{\infty}\frac{\left(\frac{1}{2}\right)_{k}\left(\frac{1}{2}\right)_{k}}{\left(\frac{3}{2}
\right)_{k}k!}(-\eta)^{k}\\
&\sim\left(2N+\alpha\right)\sqrt{\eta},
\end{split}
\end{equation}
since $|\eta|\ll1$, for large $N$ and here, the Pochhammer symbol (also called the shifted factorial)
\begin{equation*}
(x)_{k}:=\frac{\Gamma\left(k+x\right)}{\Gamma(x)}=x(x+1)\cdots(x+k-1).
\end{equation*}

For $|\eta|<1$, $\sqrt{1+\eta}$, applying the binomial theorem, we get
\begin{equation}\label{0.9}
\sqrt{1+\eta}=\frac{1}{\Gamma\left(-\frac{1}{2}\right)}\sum_{k=0}^{\infty}(-1)^{k}\frac{\Gamma\left(k-\frac{1}{2}\right)}{\Gamma\left(k+1\right)}\eta^{k},
\end{equation}
we see that
\begin{equation}\label{x}
\begin{split}
\left[\frac{b-a}{2}-\frac{(b-a)t}{2z\sqrt{ab}}\right]\sqrt{\eta(\eta+1)}&=\left[\frac{b-a}{2}-\frac{(b-a)t}{2z\sqrt{ab}}\right]\frac{\sqrt{\eta}}{\Gamma\left(-\frac{1}{2}\right)}\sum_{k=0}^{\infty}(-1)^{k}\frac{\Gamma\left(k-\frac{1}{2}\right)}{\Gamma\left(k+1\right)}\eta^{k}\\
&\sim\left[\frac{b-a}{2}-\frac{(b-a)t}{2z\sqrt{ab}}\right]\sqrt{\eta}.
\end{split}
\end{equation}
substituting (\ref{w}) and (\ref{x}) into (\ref{z}), and bear in mind $\eta=-z(b-a)^{-1}$, then we find

\begin{cor}\label{thm1}
For $N\rightarrow\infty$, the normalized polynomials associated with the weight $w_{\alpha}(x;t)$ are approximated by
\begin{equation}\label{PN}
\begin{split}
\mathcal{P}_{N}(z)\sim&\frac{(-1)^{N}}{\sqrt{2\pi}}(-z)^{-\frac{\alpha}{2}-\frac{1}{4}}(b-a)^{-\frac{1}{4}}{\rm e}^{\frac{z}{2}+\frac{t}{2z}}\\
&\cdot\exp\bigg\{\bigg[\frac{2N+\alpha}{(b-a)^{\frac{1}{2}}}
+\frac{(b-a)^{\frac{1}{2}}}{2}-\frac{(b-a)^{\frac{1}{2}}t}{2z\sqrt{ab}}\bigg](-z)^{\frac{1}{2}}\bigg\},
\end{split}
\end{equation}
where $z\notin[a,b]$.
\end{cor}
}

\begin{rem}\label{400}
Letting $\alpha=0,~ t=0$, the classical result for Laguerre polynomials due to Perron [\cite{Szego}] is recovered,
\begin{equation*}
\mathcal{P}_{N}(z)\sim\frac{(-1)^{N}}{2\sqrt{\pi}}\left(-zN\right)^{-\frac{1}{4}}\exp\left[\frac{z}{2}+2\sqrt{-zN}\right],~~z\notin[0,\infty).
\end{equation*}
\end{rem}
\begin{rem}\label{20}
To meet the demands of some proofs in the following, we first show the Laplace method here,
\begin{equation}\label{laplace1}
\int_{a}^{b}f(t){\rm e}^{-\lambda g(t)}dt\sim {\rm e}^{-\lambda g(c)}f(c)\sqrt{\frac{2\pi}{\lambda g''(c)}}~,~~~~{\rm as}~~\lambda\rightarrow\infty,
\end{equation}
where $g$ assumes a strict minimum over $[a,b]$ at an interior critical point $c$, such that
\begin{equation*}
\begin{split}
g'(c)=0,~~~~~g''(c)>0~~~~~{\rm and}~~~~~f(c)\neq0.\\
\end{split}
\end{equation*}
An alternative expression for Laplace method may be stated as follows:

If for $x\in[a,b]$, the real continuous function $g(x)$ has as its maximum the value $g(b)$, then as $N\rightarrow\infty$
\begin{equation}\label{laplace2}
\int_{a}^{b}f(x){\rm e}^{Ng(x)}dx\sim\frac{f(b){\rm e}^{Ng(b)}}{Ng'(b)}.
\end{equation}
\end{rem}

\section{The asymptotic behavior of $\lambda_{N}$ }
\begin{thm}
The smallest eigenvalue $\lambda_{N}$ of the $\mathcal{H}_{N}$ can be approximated by
\begin{equation*}
\begin{split}
\lambda_{N}\sim8\pi^{\frac{3}{2}}\left[(4N+2\alpha)^{\frac{1}{2}}+\frac{t}{2\sqrt{a_{N}}}-2t\right]^{\frac{1}{2}}\exp\left[1+t-2(4N+2\alpha)^{\frac{1}{2}}-\frac{t}{\sqrt{a_{N}}}\right],
\end{split}
\end{equation*}
where
\begin{equation*}
\begin{split}
a_{N}\sim\frac{t^\frac{2}{3}}{2(2N+\alpha)^{\frac{1}{3}}}+\frac{\alpha t^\frac{1}{3}}{3(2N+\alpha)^{\frac{2}{3}}}+\frac{\alpha^2}{6(2N+\alpha)}+\frac{5\alpha^{3}}{81t^{\frac{1}{3}}(2N+\alpha)^{\frac{4}{3}}}.
\end{split}
\end{equation*}
\end{thm}
\begin{proof}
{\color{black}Indeed, with $P_N(z)$ having the form (\ref{PN}), we observe that for large enough $\mu,\nu$ the essential contribution to $\mathcal{K}_{\mu\nu}$ comes from the arc of the unit circle around $z=-1$.} Let $\omega>0$ be a fixed number and restrict the values of $\mu$ and $\nu$ to satisfy
\begin{equation}\label{uv}
\begin{split}
N-\omega N^{\frac{1}{2}}\leq\mu,\nu\leq N,~~N\rightarrow\infty,
\end{split}
\end{equation}
thus we have
{\color{black}\begin{equation*}
\begin{split}
\mathcal{K}_{\mu\nu}&=\int_{-\pi}^{\pi}\mathcal{P}_{\mu}\left({\rm e}^{{\rm i}\phi}\right)\mathcal{P}_{\nu}\left({\rm e}^{-{\rm i}\phi}\right)d\phi\\
&\sim\int_{\pi-\varepsilon}^{\pi+\varepsilon}\mathcal{P}_{\mu}\left({\rm e}^{{\rm i}\phi}\right)\mathcal{P}_{\nu}\left({\rm e}^{-{\rm i}\phi}\right)d\phi\\
&\sim\int_{-\varepsilon}^{\varepsilon}\mathcal{P}_{\mu}\left(-{\rm e}^{{\rm i}\theta}\right)\mathcal{P}_{\nu}\left(-{\rm e}^{-{\rm i}\theta}\right)d\theta\\
\end{split}
\end{equation*}}
Expanding the integrand for $|\theta|\ll1$ with $z=-{\rm e}^{{\rm i}\theta}$, following from (\ref{PN}), we obtain
{\small\begin{equation*}
\begin{split}
\mathcal{K}_{\mu\nu}&\sim\frac{(-1)^{\mu+\nu}}{2\pi}(b_{\mu}-a_{\mu})^{-\frac{1}{4}}(b_{\nu}-a_{\nu})^{-\frac{1}{4}}\int_{-\varepsilon}^{\varepsilon}\exp\bigg\{\left(1-\frac{\theta^{2}}{8}\right)
\Big[(4\mu+2\alpha)^{\frac{1}{2}}-\frac{t}{2z\sqrt{a_{N}}}{\color{black}+\frac{z}{2}}\\
&+\frac{t}{2z}+(4\nu+2\alpha)^{\frac{1}{2}}-\frac{t}{2\overline{z}\sqrt{a_{N}}}{\color{black}+\overline{\frac{z}{2}}}+\frac{t}{2\overline{z}}\Big]+\frac{{\rm i}\theta}{2}\left[(4\mu+2\alpha)^{\frac{1}{2}}-(4\nu+2\alpha)^{\frac{1}{2}}\right]\bigg\}d\theta~({\rm a})\\
&\sim\frac{(-1)^{\mu+\nu}}{2\pi}(4N+2\alpha)^{-\frac{1}{2}}{\rm e}^{-1-t}\exp\left[(4\mu+2\alpha)^{\frac{1}{2}}+(4\nu+2\alpha)^{\frac{1}{2}}+\frac{t}{\sqrt{a_{N}}}\right]\\
&~~~~\cdot\int_{-\infty}^{\infty}\exp\bigg\{-\frac{\theta^{2}}{8}\left[2(4N+2\alpha)^{\frac{1}{2}}+\frac{t}{\sqrt{a_{N}}}-{\color{black}t}\right]\bigg\}d\theta~~~~~~~~~
~~~~~~~~~~~~~~~~~~~~~({\rm b})\\
&\sim (-1)^{\mu+\nu}\pi^{-\frac{1}{2}}(4N+2\alpha)^{-\frac{1}{2}}\left[(4N+2\alpha)^{\frac{1}{2}}+\frac{t}{2\sqrt{a_{N}}}-{\color{black}\frac{t}{2}}\right]^{-\frac{1}{2}}{\rm e}^{-1-t+\frac{t}{\sqrt{a_{N}}}}\\
&~~~~~\cdot\exp\left[(4\mu+2\alpha)^{\frac{1}{2}}+(4\nu+2\alpha)^{\frac{1}{2}}\right].~~~~~~~~~~~~~~~~~~~~~~~~~~~~~~~~~~~~~~~~~~~~~~~~~~~({\rm c})
\end{split}
\end{equation*}}
Note that the term $(4\mu+2\alpha)^{\frac{1}{2}}-(4\nu+2\alpha)^{\frac{1}{2}}$ where in (a) remains bounded because of restricting $\mu$ and $\nu$ as in (\ref{uv}), so we can get rid of the linear term in the (a) for $\theta\ll1$. As previously mentioned, contributions to the integral (a) from $(-\infty,\varepsilon)$ and $(\varepsilon,\infty)$ are small enough compared with those from $[-\varepsilon,\varepsilon]$ as $\mu\rightarrow\infty$ and $\nu\rightarrow\infty$. Therefore, we can extend the integration interval to $(-\infty,\infty)$ but without affecting the approximation of $\mathcal{K}_{\mu\nu}$, so (b) holds. Using the laplace method given by (\ref{laplace1}), we obtain (c). Observing (c), we can find that when $\mu$ and $\nu$ are confined by (\ref{uv}) and large enough,
\begin{equation}\label{0.7}
\mathcal{K}_{\mu\nu}\sim(-1)^{\mu+\nu}\mathcal{K}_{\mu\mu}^{\frac{1}{2}}\mathcal{K}_{\nu\nu}^{\frac{1}{2}}.
\end{equation}

Using the approach of [\cite{C3}] and [\cite{C9}] with the $\left\{\xi_{\mu}\right\}$ vectors, allows us to determine the asymptotic behavior of $\lambda_{N}$ for large $N$, as follows
\begin{equation*}
\xi_{\mu}=
\begin{cases}
(-1)^{\mu}\sigma\mathcal{K}_{\mu\mu}^{\frac{1}{2}}, ~{\rm if} ~N_{0}\leq\mu\leq N,\ \ \\\\
0,\ \
~~~~~~~~~~~~~{\rm if}~\mu<N_{0},
\end{cases}
\end{equation*}
where $N_{0}:=E\left[N-\omega N^{\frac{1}{2}}\right]$, whilst the positive number $\sigma$ is determined depending on the condition
\begin{equation}\label{0.8}
\sum_{\mu=0}^{N}|\xi_{\mu}|^{2}=\sigma^{2}\sum_{\mu=N_{0}}^{N}\mathcal{K}_{\mu\mu}=1.
\end{equation}
It follows from (\ref{0.7}) and (\ref{0.8}) that
\begin{equation}\label{40}
\begin{split}
\sum_{\mu,\nu=0}^{N}\mathcal{K}_{\mu\nu}\xi_{\mu}\overline{\xi}_{\nu}=\sum_{\mu,\nu=N_{0}}^{N}(-1)^{\mu+\nu}\sigma^{2}\mathcal{K}_{\mu\nu}\mathcal{K}_{\mu\mu}^{\frac{1}{2}}\mathcal{K}_{\nu\nu}^{\frac{1}{2}}
\sim\sigma^{2}\left(\sum_{\mu=N_{0}}^{N}\mathcal{K}_{\mu\mu}\right)^{2}
=\sum_{\mu=N_{0}}^{N}\mathcal{K}_{\mu\mu}.
\end{split}
\end{equation}
This means the minimum value in equation (\ref{e3}) can be approximated by (\ref{e4}), following (\ref{40}), because of the arbitrariness of $\omega$, i.e.
\begin{equation*}
\lambda_{N}\sim\frac{2\pi}{\sum_{\mu=0}^{N}\mathcal{K}_{\mu\mu}}.
\end{equation*}
It follows that
\begin{equation}\label{5000}
\lambda_{N}\sim\frac{2\pi}{\int_{0}^{N}\mathcal{K}_{\mu\mu}d\mu}.
\end{equation}
Hence we get
\begin{equation*}
\begin{split}
\frac{2\pi}{\lambda_{N}}&\sim\sum_{\mu=0}^{N}\mathcal{K}_{\mu\mu}\\
&\sim\pi^{-\frac{1}{2}}(4N+2\alpha)^{-\frac{1}{2}}\left[(4N+2\alpha)^{\frac{1}{2}}+\frac{t}{2\sqrt{a_{N}}}-{\color{black}\frac{t}{2}}\right]^{-\frac{1}{2}}{\rm e}^{-1-t+\frac{t}{\sqrt{a_{N}}}}\int_{0}^{N}\exp\left[2(4x+2\alpha)^{\frac{1}{2}}\right]dx\\
&\sim4^{-1}\pi^{-\frac{1}{2}}\left[(4N+2\alpha)^{\frac{1}{2}}+\frac{t}{2\sqrt{a_{N}}}-{\color{black}\frac{t}{2}}\right]^{-\frac{1}{2}}{\rm e}^{-1-t+\frac{t}{\sqrt{a_{N}}}}\exp\left[2(4N+2\alpha)^{\frac{1}{2}}\right].\\
\end{split}
\end{equation*}
Consequently,
\begin{equation*}
\begin{split}
\lambda_{N}\sim8\pi^{\frac{3}{2}}\left[(4N+2\alpha)^{\frac{1}{2}}+\frac{t}{2\sqrt{a_{N}}}-{\color{black}\frac{t}{2}}\right]^{\frac{1}{2}}{\rm e}^{1+t-\frac{t}{\sqrt{a_{N}}}}\exp\left[-2(4N+2\alpha)^{\frac{1}{2}}\right].\\
\end{split}
\end{equation*}
which complete the proof.
\end{proof}

\begin{rem}
For the classical Laguerre weight $w_{0}(x):=x^{\alpha}{\rm e}^{-x},~x\in[0,\infty),~\alpha>-1$, {\rm i.e.} taking $t=0$ for our weight $w(x)$, we have
\begin{equation*}
\lambda_{N}\sim2^{\frac{13}{4}}\pi^{\frac{3}{2}}{\rm e}\left(2N+\alpha\right)^{\frac{1}{4}}\exp\left[-2^{\frac{3}{2}}\left(2N+\alpha\right)^{\frac{1}{2}}\right].
\end{equation*}
which agrees with the results from [\cite{MMAS}].
\end{rem}

\begin{rem}
Noting that when $\alpha=0,~t=0$, Szeg\"{o}'s [\cite{C3}] classical result for the Laguerre weight $w_{1}(x):={\rm e}^{-x}$ is recovered:
\begin{equation*}
\begin{split}
&\lambda_{N}\sim2^{\frac{7}{2}}\pi^{\frac{3}{2}}{\rm e}N^{\frac{1}{4}}\exp\left[-4N^{\frac{1}{2}}\right].
\end{split}
\end{equation*}
\end{rem}

{\color{black}\section{Appendix: Integral identities}
\subsection{Appendix A}
The integrals identities listed below, which are relevant to our derivation and can be found in [\cite{new1,C21,new2}].

\begin{equation*}
~~~~~~~~~~~~~~~~~~~~~~~~~~~~~~~~~~~~~~~~~~\int_{a}^{b}\frac{dx}{\sqrt{(b-x)(x-a)}}=\pi.~~~~~~~~~~~~~~~~~~~~~~~~~~~~~~~~~~~~~~~~~~~~~~~({\color{black}\rm A1})
\end{equation*}
\begin{equation*}
~~~~~~~~~~~~~~~~~~~~~~~~~~~~~~~~~~~~~~~~~~~~~~~\int_{a}^{b}\frac{xdx}{\sqrt{(b-x)(x-a)}}=\pi\frac{a+b}{2}.~~~~~~~~~~~~~~~~~~~~~~~~~~~~~~~~~~~~~~~~~~({\color{black}\rm A2})
\end{equation*}
\begin{equation*}
~~~~~~~~~~~~~~~~~~~~~~~~~~~~~~~~~~~~~~~~~~~\int_{a}^{b}\frac{dx}{x^{2}\sqrt{(b-x)(x-a)}}=
\frac{(a+b)\pi}{2(ab)^{\frac{3}{2}}}.~~~~~~~~~~~~~~~~~~~~~~~~~~~~~~~~~~~~~({\color{black}\rm A3})
\end{equation*}
\begin{equation*}
~~~~~~~~~~~~~~~~~~~~~~~~~~~~~~~~~~~\int_{a}^{b}\frac{dx}{(x+t)\sqrt{(b-x)(x-a)}}=
\frac{\pi}{\sqrt{(t+a)(t+b)}}.~~~~~~~~~~~~~~~~~~~~~~({\color{black}\rm A4})
\end{equation*}
\begin{equation*}
~~~~~~~~~~~\int_{a}^{b}\frac{\log(x+t)dx}{x\sqrt{(b-x)(x-a)}}=
\frac{\pi}{\sqrt{ab}}\log\left[\frac{\left(\sqrt{ab}+\sqrt{(t+a)(t+b)}\right)^{2}-t^{2}}{\left(\sqrt{a}+\sqrt{b}\right)^{2}}\right].~~~~~({\color{black}\rm A5})
\end{equation*}
\subsection{Appendix B}
In this section, we will prove the following equation:
\begin{equation*}
~~~~~~~~~~~~~~~~\int_{a}^{b}\frac{dx}{x(x-t)\sqrt{(b-x)(x-a)}}
=-\frac{\pi}{t}\left[\frac{1}{\sqrt{ab}}+\frac{1}{\sqrt{(t-a)(t-b)}}\right].~~~~~~~~~~~~({\color{black}\rm B1})
\end{equation*}

Defining
\begin{equation*}
F(\zeta):=\frac{1}{\zeta(\zeta-t)\sqrt{(\zeta-a)(\zeta-b)}},~~\zeta\in \Lambda,~~t\notin[a,b],
\end{equation*}
where the definition of $\Lambda$ can be found in the below figure:
\begin{figure}[H]
\centering
\includegraphics[width=0.4\textwidth]{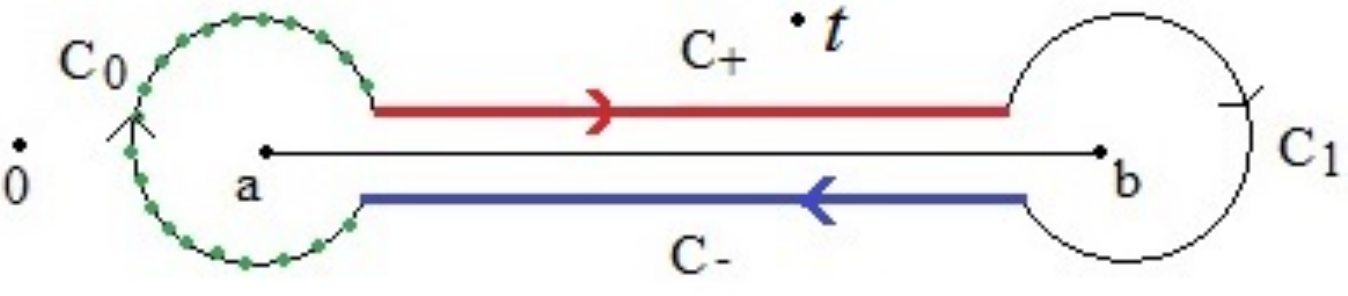}~~~~\text{where $\Lambda:=C_{0}+C_{+}+C_{1}+C_{-}$,}
\caption{The integration path $\Lambda$ of $F(\zeta)$. }
\end{figure}
then
\begin{equation*}
\int_{\Lambda}F(\zeta)d\zeta=\int_{\Lambda}\frac{1}{\zeta(\zeta-t)\sqrt{(\zeta-a)(\zeta-b)}}d\zeta=
\left(\int_{C_{0}}+\int_{C_{+}}+\int_{C_{1}}+\int_{C_{-}}\right)F(\zeta)d\zeta
\end{equation*}
where $\zeta$ is in the interior of $\Lambda$. When $\zeta\rightarrow x\in(a,b)$, $C_{0}$ tends to a closed circle centered at $a$ with radius $x-a$, similarly, $C_{1}$ tends to a closed circle centered at $b$ with radius $b-x$. Thus by Cauchy-Gourat theorem,
\begin{equation*}
\int_{C_{0}}F(\zeta)d\zeta\rightarrow0,~~~\int_{C_{1}}F(\zeta)d\zeta\rightarrow0.
\end{equation*}

For the case of $C_{+}$, which is shown as below,
\begin{figure}[H]
\centering
\includegraphics[width=0.4\textwidth]{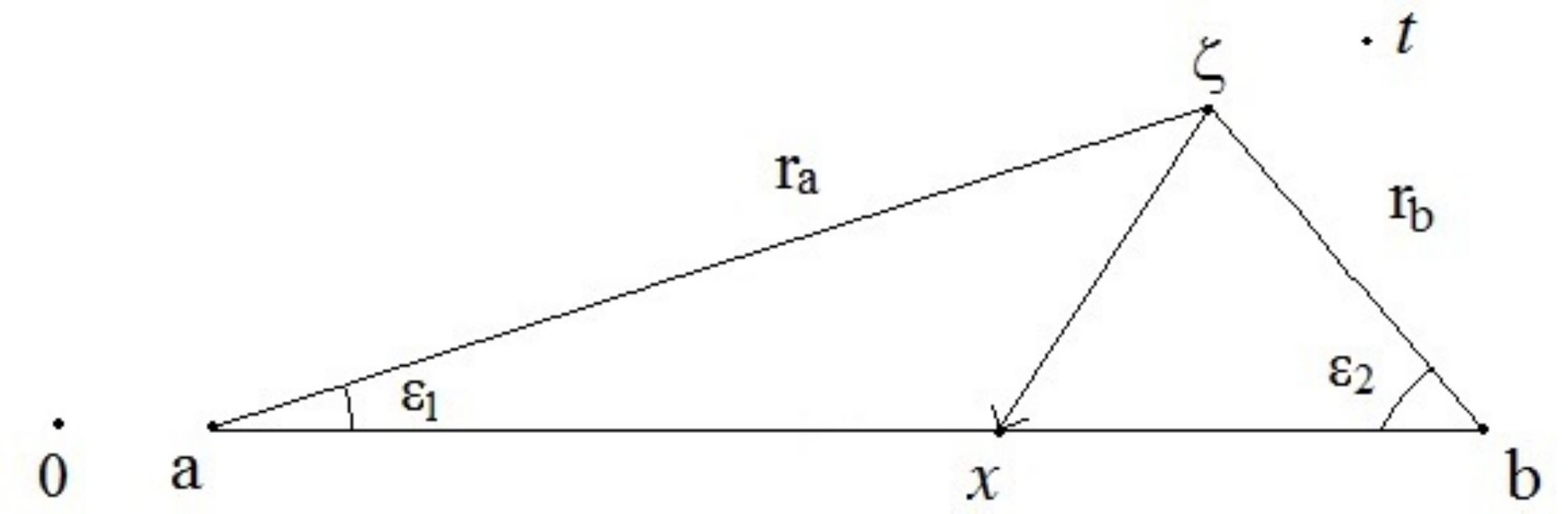}
\caption{The integration path for $\zeta$ over $(a,b)$. }
\end{figure}
where $\zeta-a=r_{a}{\rm e}^{{\rm i}\varepsilon_{1}}$, $\zeta-b=r_{b}{\rm e}^{{\rm i}(\pi-\varepsilon_{2})}$, then as $\zeta\rightarrow x$ and $\varepsilon_{1},\varepsilon_{2}\rightarrow 0$, we have
\begin{equation*}
\int_{C_{+}}F(\zeta)d\zeta=\int_{a}^{b}\frac{{\color{black}1}}{\zeta(\zeta-t)\sqrt{r_{a}r_{b}{\rm e}^{{\rm i}(\pi-\varepsilon_{2}+\varepsilon_{1})}}}d\zeta
\rightarrow\int_{a}^{b}\frac{dx}{{\rm i}x(x-t)\sqrt{(b-x)(x-a)}}.
\end{equation*}
In the same way, for the case of $C_{-}$,
\begin{figure}[H]
\centering
\includegraphics[width=0.4\textwidth]{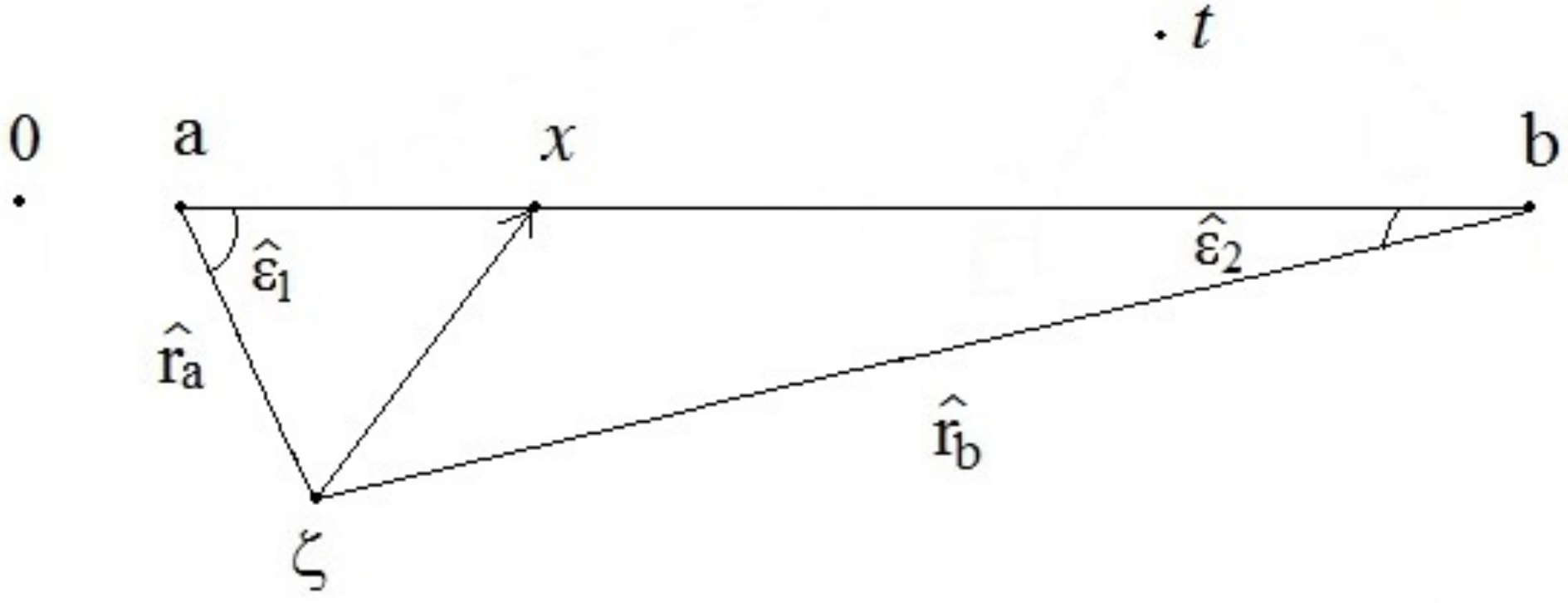}
\caption{The integration path for $\zeta$ below $(a,b)$. }
\end{figure}
where $\zeta-a=\widehat{r_{a}}{\rm e}^{{\rm i}(2\pi-\widehat{\varepsilon}_{1})}$, $\zeta-b=\widehat{r_{b}}{\rm e}^{{\rm i}(\pi+\widehat{\varepsilon}_{2})}$, then as $\zeta\rightarrow x$ and $\widehat{\varepsilon}_{1},\widehat{\varepsilon}_{2}\rightarrow 0$, we obtain
\begin{equation*}
\int_{C_{-}}F(\zeta)d\zeta=\int_{b}^{a}\frac{d\zeta}{\zeta(\zeta-t)\sqrt{\widehat{r_{a}}\widehat{r_{b}}{\rm e}^{{\rm i}(3\pi+\widehat{\varepsilon}_{2}-\widehat{\varepsilon}_{1})}}}d\zeta
\rightarrow\int_{a}^{b}\frac{dx}{{\rm i}x(x-t)\sqrt{(b-x)(x-a)}}.
\end{equation*}
Hence,
\begin{equation*}
\int_{\Lambda}F(\zeta)d\zeta\rightarrow\frac{2}{\rm i}\int_{a}^{b}\frac{dx}{x(x-t)\sqrt{(b-x)(x-a)}}.
\end{equation*}
We note that $F(\zeta)$ has {\color{black}two poles $\zeta=0$, $\zeta=t$ and and a removable singularity $\zeta=\infty$}, then
\begin{equation*}
\begin{split}
\int_{\Lambda}F(\zeta)d\zeta&=2\pi{\rm i}\mathop{\rm Res}\limits_{\zeta=0} F(\zeta)+2\pi{\rm i}\mathop{\rm Res}\limits_{\zeta=t} F(\zeta)+2\pi{\rm i}\mathop{\rm Res}\limits_{\zeta=\infty}F(\zeta)\\
&=\frac{2\pi{\rm i}}{t\sqrt{ab}}+\frac{2\pi{\rm i}}{t\sqrt{(t-a)(t-b)}}-2\pi{\rm i}\mathop{\rm Res}\limits_{\eta=0}\left[F\left(\frac{1}{{\color{black}\eta}}\right)\frac{1}{{\color{black}\eta}^{2}}\right]\\
&=2\pi{\rm i}\left(\frac{1}{t\sqrt{ab}}+\frac{1}{t\sqrt{(t-a)(t-b)}}\right).
\end{split}
\end{equation*}
Consequently,
\begin{equation*}
\int_{a}^{b}\frac{dx}{x(x-t)\sqrt{(b-x)(x-a)}}=\frac{\rm i}{2}\int_{\Lambda}F(\zeta)d\zeta=-\frac{\pi}{t}\left(\frac{1}{\sqrt{ab}}+\frac{1}{\sqrt{(t-a)(t-b)}}\right).
\end{equation*}}

\section{Acknowledgements}
 M. Zhu and Y. Chen would like to thank the Science and Technology Development Fund of the Macau SAR for generous support in providing FDCT 023/2017/A1. They would also like to thank the University of Macau for generous support via MYRG 2018-00125 FST. {\color{black}C. Li acknowledges the support  of the National Natural Science Foundation of China under Grant no. 11571192 and K. C. Wong Magna Fund in Ningbo University.}

\end{document}